\documentclass[12pt]{article}

\usepackage{amsmath}
\usepackage{amssymb}
\usepackage{amsfonts}
\usepackage{latexsym}
\usepackage{color}

\catcode `\@=11 \@addtoreset{equation}{section}

\catcode `\@=12

\newcommand{\be}{\begin{equation}}
\newcommand{\en}{\end{equation}}
\newcommand{\bea}{\begin{eqnarray}}
\newcommand{\ena}{\end{eqnarray}}
\newcommand{\beano}{\begin{eqnarray*}}
\newcommand{\enano}{\end{eqnarray*}}
\newcommand{\bee}{\begin{enumerate}}
\newcommand{\ene}{\end{enumerate}}

\newcommand{\mc}{\mathcal}

\newcommand{\U}{{\mc U}}
\newcommand{\V}{{\mc V}}
\newcommand{\Pc}{{\mc P}}

\newcommand{\E}{{\cal E}}

\newcommand{\G}{{\cal G}}

\newcommand{\Lc}{{\cal L}}
\newcommand{\C}{{\cal C}}
\newcommand{\T}{{\cal T}}
\newcommand{\1}{1 \!\! 1}

\newcommand{\Hil}{\mc H}

\newtheorem{thm}{Theorem}
\newtheorem{cor}[thm]{Corollary}

\newtheorem{lem}[thm]{Lemma}
\newtheorem{prop}[thm]{Proposition}
\newtheorem{defn}[thm]{Definition}

\newenvironment{proof}{\noindent {\bf Proof --}}{\hfill$\square$ \vspace{3mm}\endtrivlist}

\catcode `\@=11 \@addtoreset{equation}{section}
\catcode `\@=12

\begin{document}

\thispagestyle{empty}

\vspace*{2cm}

\begin{center}
{\Large \bf On some properties of g-frames and g-coherent states}   \vspace{2cm}\\

{\large M.R. Abdollahpour}\\
  Department of mathematics Faculty of sciences, University of Mohaghegh Ardabili, Ardabil, Iran\\
e-mail: mrabdollahpour@yahoo.com \vspace{5mm}\\

{\large F. Bagarello}\\
  Dipartimento di Metodi e Modelli Matematici,
Facolt\`a di Ingegneria,\\ Universit\`a di Palermo, I-90128  Palermo, Italy\\
e-mail: bagarell@unipa.it

\end{center}

\vspace*{2cm}

\begin{abstract}
\noindent After a short review of some basic facts on g-frames, we analyze in  details the so-called (alternate) dual g-frames.  We end the paper by introducing what we call {\em g-coherent states} and studying their  properties.

\end{abstract}

\vspace{2cm}



{\bf Running title}:  g-frames and g-coherent states.

\vfill


\newpage

\section{Introduction}

In a series of recent papers, \cite{sun,na,sun2} and references therein, a class of bounded operators, the so called {\em g-frames}, has been introduced and studied in some details. These operators allow the extension of the notion of {\em standard frame}, and this explains the "g" in their name which stands for {\em generalized}. Then one of us (MRA) has introduced a particular class of g-frames, focusing his attention mainly on some mathematical aspects of these operators, \cite{arf}. The second author (FB)  used g-frames to construct examples of physical systems which are of a certain interest in supersymmetric quantum mechanics, \cite{fb}. Here we begin our joint analysis on this subject.

 The paper is organized as follows: in the next section we introduce g-frames and review some of the standard results, adopting a terminology which is sometimes slightly different from the usual one since we believe it simplifies the notation and some proofs.  Then we show that the  dual of a given g-frame is not, in general, unique, but it becomes unique under suitable extra conditions on the original set. The last section is devoted to the introduction of what we call {\em g-coherent states}, which will appear to be essentially a two-dimensional version of the standard coherent states. In particular we will deduce several resolutions of the identity associated to different g-coherent states.

\section{Description of the system}

Let $\Hil$ be a given Hilbert space with scalar product $\left<.,.\right>_\Hil$ and related norm $\|.\|_\Hil$, and $\tilde\Hil$ a second Hilbert space with scalar product $\left<.,.\right>_{\tilde\Hil}$ and  norm $\|.\|_{\tilde\Hil}$. Let now $J$ be a set of indexes which labels a  sequence of Hilbert spaces $\{\tilde\Hil_j\subseteq\tilde\Hil,\,j\in J\}$. We call $\left<.,.\right>_{\tilde\Hil_j}$ and  $\|.\|_{\tilde\Hil_j}$ their scalar products and norms.

\begin{defn}\label{def1}
A set of bounded operators $\Lc=\{\Lambda_j\in B(\Hil,\tilde\Hil_j), \,j\in J\}$ is an $(A,B)$ g-frame of $(\Hil,\{\tilde\Hil_j\})$, briefly a g-frame, if there exist two positive numbers $A$ and $B$, with $0<A\leq B<\infty$, such that, for all $f\in\Hil$,
\be
 A\,\|f\|_\Hil^2 \leq \sum_{j\in J}\|\Lambda_jf\|_{\tilde\Hil_j}^2 \leq B\,\|f\|_\Hil^2
\label{21}
\en

\end{defn}

In particular a g-frame is called {\em tight} if $A=B$ and it is called a {\em Parseval g-frame} if $A=B=1$. Because of the several different Hilbert spaces appearing in this paper, we will use different symbols to indicate the different norms and scalar products which appear in these different spaces. Incidentally, notice that standard frames are recovered when $\tilde\Hil_j=\tilde\Hil=\Bbb{C}$ for all $j$ and $\Lambda_j=<\varphi_j,.>$, with $\varphi_j$ belonging to a certain $(A,B)$-frame of $\Hil$. Notice also that, for all $\tilde f\in\tilde\Hil_j\subseteq\tilde\Hil$, we have $\|\tilde f\|_{\tilde\Hil_j}=\|\tilde f\|_{\tilde\Hil}$.

Since the adjoint of $\Lambda_j$, $\Lambda_j^\dagger$, is bounded from $\tilde\Hil_j$ into $\Hil$, it follows that $\Lambda_j^\dagger\Lambda_j\in B(\Hil)$ for all $j\in J$. Now, let us define yet another Hilbert space, $\hat\Hil$, which looks like an  $l^2(\Bbb{N})$ space but in which the sequences of complex numbers
are replaced by sequences of elements of the different $\tilde\Hil_j$'s. More explicitly, let
\be
\hat\Hil:=\left\{\underline{f}:=\{\tilde f_j\in\tilde\Hil_j\}_{j\in J}, \mbox{ such that }\|\underline{f}\|_{\hat\Hil}^2:=
\sum_{j\in J}\|\tilde f_j\|_{\tilde\Hil_j}^2<\infty\right\},
\label{22}\en
which we endow with the following scalar product:
\be
<\underline{f},\underline{g}>_{\hat\Hil}:=\sum_{j\in J}<\tilde f_j,\tilde g_j>_{\tilde\Hil_j}.
\label{23}\en
Now we can associate to the set $\Lc$  a bounded operator $T_\Lc:\Hil \rightarrow \hat\Hil$,  the {\em analysis operator},  defined as follows:
\be
\forall f\in \Hil \hspace{1cm}  (T_\Lc\,f)_j = \Lambda_j\, f, \qquad\Longrightarrow\qquad T_\Lc\,f=\{\Lambda_j\, f\}_{j\in J}.
\label{24}
\en
The vector $(T_\Lc\,f)_j$ belongs to $\tilde\Hil_j$ for each $j\in J$, while $T_\Lc\,f$ belongs to $\hat\Hil$. As for standard frames, we find that $\|T_\Lc\|_{B(\Hil,\hat\Hil)} \leq \sqrt{B}$. The
adjoint of  $T_\Lc$, the so-called {\em synthesis operator} $T_\Lc^\dagger$,  maps $\hat\Hil$ into $\Hil$, and acts on a generic vector $\underline{f}$ of $\hat\Hil$ as follows:
 \be
   T_\Lc^\dagger \underline{f} = \sum_{j\in J}\Lambda_j^\dagger \,f_j.
\label{25}
\en
Of course $\|T_\Lc^\dagger\|_{B(\hat\Hil,\Hil)} \leq \sqrt{B}$. Using these two operators we can define the g-frame operator $S_\Lc=T_\Lc^\dagger\,T_\Lc$ which acts on a generic element $f\in\Hil$ as
\be
S_\Lc\,f=T_\Lc^\dagger\,T_\Lc f=\sum_{j\in J}\Lambda_j^\dagger\,\Lambda_j\,f, \quad\mbox{ so that }\quad S_\Lc=T_\Lc^\dagger\,T_\Lc=\sum_{j\in J}\Lambda_j^\dagger\,\Lambda_j.
\label{26}\en
Obviously $S_\Lc:\Hil\rightarrow\Hil$ and we have $\|S_\Lc\|_{B(\Hil)}\leq B$. Moreover $S_\Lc$ is self-adjoint and strictly positive. Also, using (\ref{26}), we can restate Definition \ref{def1} in the following equivalent form:

{\em $\Lc$ is an (A,B) {g-frame} of $(\Hil,\{\tilde\Hil_j\})$ if there exist
two positive numbers $A$ and $B$, $0<A\leq B<\infty$, such that the inequalities
\be
A\1_\Hil \leq S_\Lc \leq B \1_\Hil
\label{27}
\en
hold in the sense of the operators.}

 Hence we find that $A\leq\|S_\Lc\|_{B(\Hil)}\leq B$, so that the norm of $S_\Lc$ is also bounded from below. $S_\Lc^{-1}$ clearly exists in $\Hil$, and we find that $B^{-1}\1_\Hil \leq S_\Lc^{-1} \leq A^{-1} \1_\Hil$. $S_\Lc^{-1}$ and $S_\Lc$ can be used together now to get two {\em resolutions of the identity} in $\Hil$. Indeed, defining a new operator $\tilde\Lambda_j:=\Lambda_jS_\Lc^{-1}$, which maps again $\Hil$ into $\tilde\Hil_j$, and its adjoint $\tilde\Lambda_j^\dagger:=S_\Lc^{-1}\Lambda_j^\dagger:\tilde\Hil_j\rightarrow\Hil$, we find that
\be
f=S_\Lc\,S_\Lc^{-1}f=\sum_{j\in J}\Lambda_j^\dagger\,\tilde\Lambda_j \,f
\quad\mbox{ or }\quad f=S_\Lc^{-1}S_\Lc\,f=\sum_{j\in J}\tilde\Lambda_j^\dagger\,\Lambda_j \,f.
\label{28}\en
In an operatorial form we can rewrite (\ref{28}) as
\be
\sum_{j\in J}\Lambda_j^\dagger\,\tilde\Lambda_j =
\sum_{j\in J}\tilde\Lambda_j^\dagger\,\Lambda_j =\1_\Hil,
\label{29}\en
which are the  resolutions of the identity we were looking for, which are related since the second is just the adjoint of the first one. The {\em canonical dual set} of $\Lc$, $\tilde\Lc=\{\tilde\Lambda_j=\Lambda_j\,S_\Lc^{-1}\in B(\Hil,\tilde \Hil_j),\,j\in J\}$, is a g-frame by itself, and in particular is a $\left(\frac{1}{B},\frac{1}{A}\right)$ g-frame, whose canonical dual, $\tilde{\tilde\Lc}$, coincides with $\Lc$ itself. It is sometimes useful to rewrite (\ref{29}) in terms of the synthesis and analysis operators of $\Lc$ and $\tilde\Lc$:
\be
T_{\Lc}^\dagger\,T_{\tilde\Lc}=T_{\tilde\Lc}^\dagger T_{\Lc}=\1_\Hil.
\label{29bis}\en

As for standard frames we can check that the set $Q:=\{Q_j:=\Lambda_j\,S^{-1/2}\in B(\Hil,\tilde\Hil_j),\,j\in J\}$ is a Parseval g-frame, for any starting g-frame $\Lc=\{\Lambda_j\in B(\Hil,\tilde \Hil_j),\,j\in J\}$.

\vspace{2mm}

In the literature, \cite{sun}, we can find the following useful definitions on g-frames: let $\Hil$, $\tilde\Hil_j$ and $\tilde\Hil$ be as before. Then:

\begin{defn}\label{def2b}
A set of operators $\Lc=\{\Lambda_j\in B(\Hil,\tilde\Hil_j), \,j\in J\}$ is g-orthonormal (g-on) in $(\Hil,\{\tilde\Hil_j\})$ if, for all $\tilde f_j\in\tilde\Hil_j$, $\tilde f_k\in\tilde\Hil_k$,
\be
\left<\Lambda_j^\dagger\tilde f_j,\Lambda_k^\dagger\tilde f_k\right>_\Hil=\delta_{j,k}\, \left<\tilde f_j,\tilde f_k\right>_{\tilde\Hil}.
\label{210}\en
A g-on set $\Lc$ is a g-on basis  in $(\Hil,\{\tilde\Hil_j\})$ if, for all $f\in\Hil$, the Parseval equality
\be
\sum_{j\in J}\|\Lambda_j\,f\|_{\tilde\Hil_j}^2=\|f\|^2_\Hil
\label{211}\en

holds.
\end{defn}

This last equality implies that for a g-on basis $\Lc$ the frame operator $S_\Lc$ is just the identity in $\Hil$: $S_\Lc=\sum_{j\in J} \Lambda_j^\dagger\,\Lambda_j=T_\Lc^\dagger T_\Lc=\1_\Hil$. Of course, due to (\ref{211}), if $\Lambda_jf=0$ for all $j\in J$, then $f=0$. This is  called {\em g-completeness} of the set $\Lc$. It is  interesting to notice that, if $\Lc$ is a g-on basis, we also find
\be
T_\Lc\, T_\Lc^\dagger=\1_{\hat\Hil}.
\label{211bis}\en
Indeed we have, taking $\underline f$ and $\underline g$ in $\hat\Hil$, $$\left<\underline f, T_\Lc\, T_\Lc^\dagger\,\underline g\right>_{\hat\Hil}=\left<T_\Lc^\dagger\,\underline f,  T_\Lc^\dagger\,\underline g\right>_{\hat\Hil}=\sum_{j,k}\left<\Lambda_j^\dagger\,\tilde f_j,  \Lambda_k^\dagger\,\tilde f_k\right>_{\hat\Hil}=
\sum_{j,k}\delta_{j,k}\left<\tilde f_j, \tilde f_k\right>_{\tilde\Hil}=\left<\underline f, \underline g\right>_{\hat\Hil}.$$

Now we define a g-Riesz basis in the following way:

\begin{defn}\label{def2c}

A set of operators $\Lc=\{\Lambda_j\in B(\Hil,\tilde\Hil_j), \,j\in J\}$ is a g-Riesz basis of $(\Hil,\{\tilde\Hil_j\})$ if there exists a bounded operator $X\in B(\Hil)$ with bounded inverse $X^{-1}\in B(\Hil)$ and a g-on basis $\Theta=\{\theta_j\in B(\Hil,\tilde\Hil_j), \,j\in J\}$ in $(\Hil,\{\tilde\Hil_j\})$ such that
\be
\Lambda_j=\theta_j\,X
\label{212}\en
for all $j\in J$.
\end{defn}

In \cite{sun} it is proven that this definition is equivalent to the following one, which will be used in the proof of Proposition 13 below:

{\em
A set of operators $\Lc=\{\Lambda_j\in B(\Hil,\tilde\Hil_j), \,j\in J\}$ is a g-Riesz basis of $(\Hil,\{\tilde\Hil_j\})$ if is g-complete and if there are two positive constants $A, B>0$ such that, for all finite subset $I\subseteq J$ and for all $g_j\in\tilde\Hil_j$, we have
$$
A\sum_{j\in I}\|g_j\|_{\tilde\Hil_j}^2\leq \left\|\sum_{j\in I}\Lambda_j^\dagger g_j\right\|_\Hil\leq B\sum_{j\in I}\|g_j\|_{\tilde\Hil_j}^2
$$
}

Using Definition \ref{def2c} it is very easy to check that any g-Riesz basis in $(\Hil,\{\tilde\Hil_j\})$ is a g-frame. Indeed, using (\ref{212}) in the computation of the operator $S_\Lc:=\sum_{j\in J}\Lambda_j^\dagger\,\Lambda_j$ we find
$$
S_\Lc:=\sum_{j\in J}\Lambda_j^\dagger\,\Lambda_j=X^\dagger\,\sum_{j\in J} \theta_j^\dagger\,\theta_j\,X=X^\dagger\,X,
$$
since $S_\Theta:=\sum_{j\in J} \theta_j^\dagger\,\theta_j=\1_\Hil$,  $\Theta$ being a g-on basis. Notice that we have used here also the continuity of $X$ and $X^\dagger$. From this equality we deduce that
\be
\|X^{-1}\|^{-2}\1_\Hil\leq S_\Lc \leq \|X\|^{2}\1_\Hil,
\label{213}\en
so that inequality (\ref{27}) is satisfied with $A=\|X^{-1}\|^{-2}$ and $B=\|X\|^{2}$. Hence $\Lc$ is a g-frame, as stated. Another relevant definition for us is that of biorthogonal g-frames:

\begin{defn}
Let $\Lc=\{\Lambda_j\in B(\Hil,\tilde\Hil_j), \,j\in J\}$ and $\G=\{\Gamma_j\in B(\Hil,\tilde\Hil_j), \,j\in J\}$ be two g-frames of $(\Hil,\{\tilde\Hil_j\})$. We say that they are biorthogonal if
\be
\left<\Lambda_j^\dagger\tilde f_j,\Gamma_k^\dagger\tilde f_k\right>_\Hil=\delta_{j,k}\, \left<\tilde f_j,\tilde f_k\right>_{\tilde\Hil}
\label{214}\en
for all $\tilde f_j\in\tilde\Hil_j$, $\tilde f_k\in\tilde\Hil_k$.
\end{defn}
Then we have
\begin{thm}
Let $\Lc=\{\Lambda_j\in B(\Hil,\tilde\Hil_j), \,j\in J\}$  be a g-Riesz basis of $(\Hil,\{\tilde\Hil_j\})$ and $\tilde\Lc$ his canonical dual set. Then $\tilde\Lc$ is also a g-Riesz basis and $\Lc$ and $\tilde\Lc$ are biorthogonal.

\end{thm}

To prove this theorem it is enough to use Definition \ref{def2c} and to notice that, since $\Lambda_j=\theta_j X$ for some $X\in B(\Hil)$ with bounded inverse and for some g-on basis $\Theta=\{\theta_j\in B(\Hil,\tilde\Hil_j), \,j\in J\}$, then  we get $\tilde\Lambda_j=\theta_j (X^\dagger)^{-1}$, which implies our first statement. The biorthogonality of $\Lc$ and $\tilde\Lc$ can be  explicitly checked.

We see that most of the results obtained in this section are simple extensions of well known facts in the theory of standard frames and Riesz bases, \cite{chri}. Another result which can be extended to the present settings is the following:

\begin{thm}
Let $\Lc=\{\Lambda_j\in B(\Hil,\tilde\Hil_j), \,j\in J\}$  be a g-Riesz basis of $(\Hil,\{\tilde\Hil_j\})$. Then the set $\E:=\{E_j=\Lambda_j\,S_\Lc^{-1/2}\in B(\Hil,\tilde\Hil_j), \,j\in J\}$ is a g-on basis in $(\Hil,\{\tilde\Hil_j\})$.
\end{thm}
\begin{proof}
Indeed we have, using Theorem 5,
$$
\left<E_j^\dagger\tilde f_j,E_k^\dagger\tilde f_k\right>_\Hil=\left<S_\Lc^{-1/2}\,\Lambda_j^\dagger\tilde f_j,S_\Lc^{-1/2}\,\Lambda_k^\dagger\tilde f_k\right>_\Hil=\left<S_\Lc^{-1}\,\Lambda_j^\dagger\tilde f_j,\Lambda_k^\dagger\tilde f_k\right>_\Hil=
$$
$$
=\left<\tilde\Lambda_j^\dagger\tilde f_j,\Lambda_k^\dagger\tilde f_k\right>_\Hil=\delta_{j,k}\left<\tilde f_j,\tilde f_k\right>_{\tilde\Hil},
$$
for all $\tilde f_j\in\tilde\Hil_j$, $\tilde f_k\in\tilde\Hil_k$. Moreover, since $\Lambda_j=\theta_j X$ for some $X\in B(\Hil)$ with bounded inverse and for some g-on basis $\Theta=\{\theta_j\in B(\Hil,\tilde\Hil_j), \,j\in J\}$, we also have
$$
\sum_{j\in J}\|E_jf\|_{\tilde\Hil_j}^2=\sum_{j\in J}\|\theta_jXS_\Lc^{-1/2}f\|_{\tilde\Hil_j}^2=\|XS_\Lc^{-1/2}f\|_{\Hil}^2=\left<X^\dagger\,XS_\Lc^{-1/2} f,S_\Lc^{-1/2} f\right>_\Hil=\|f\|_\Hil^2,
$$
for all $f\in\Hil$. In this derivation we have used the equality $S_\Lc=X^\dagger X$.
\end{proof}

We end this section recalling how, following \cite{sun}, g-frames (resp. g-Riesz bases or g-on bases) can be constructed starting from ordinary frames (resp. Riesz bases or on bases) in $\Hil$. The starting point is the usual set of Hilbert spaces, $\Hil$, $\tilde\Hil_j\subseteq\tilde\Hil$, $j\in J$, and an on basis of $\tilde\Hil_j$: $\E_j:=\{e_k^{(j)},\,k\in K_j\subseteq\Bbb{Z}\}$. Here $K_j$ is a certain set of indexes labeling $\tilde\Hil_j$. Hence $\left<e_k^{(j)},e_l^{(j)}\right>_{\tilde\Hil_j}=\delta_{k,l}$ for all $k,l\in K_j$, and for all fixed $j\in J$. Let further $\Lc=\{\Lambda_j\in B(\Hil,\tilde\Hil_j), \,j\in J\}$ be a given set of bounded operators. Then we construct a new set of vectors in $\Hil$ starting from $\E_j$ and $\Lc$ in the following way:
\be
\U:=\{u_k^{(j)}:=\Lambda_j^\dagger\,e_k^{(j)}\in \Hil,\, k\in K_j,\, j\in J\}
\label{215}\en
It is easy to see that: if $\U$ is an $(A,B)-$frame of $\Hil$  then $\Lc$ is an $(A,B)$ g-frame of $(\Hil,\{\tilde\Hil_j\})$. If $\U$ is an Riesz basis of $\Hil$  then $\Lc$ is  g-Riesz basis of $(\Hil,\{\tilde\Hil_j\})$. If $\U$ is an on basis of $\Hil$  then $\Lc$ is  g-on basis of $(\Hil,\{\tilde\Hil_j\})$.

The proof of these statements, originally given in \cite{sun}, immediately follows from our definitions.

\section{Dual of g-frames}

For ordinary frames it is known that the dual set of a given frame is not necessarily its "canonically conjugate" dual, except when some extra requirement is satisfied, \cite{han}. In this section we will show that similar results can be extended to g-frames. We start with the following
\begin{defn} Let $\Lc=\{\Lambda_j\in B(\Hil,\tilde\Hil_j), \,j\in J\}$ and $\Theta=\{\theta_j\in B(\Hil,\tilde\Hil_j), \,j\in J\}$ be
two $g$-frames of $(\Hil,\{\tilde\Hil_j\})$ such that $$f=\sum_{i\in
J}\theta_{i}^\dagger\Lambda_{i}f,\quad f\in \Hil.$$ Then
$\Theta$ is called an alternate dual of
$\Lc$.

Moreover, the set $\Lc$ is a g-Bessel family of $(\Hil,\{\tilde\Hil_j\})$ if a positive constant $B$ exists for which
$$
\sum_{j\in J}\|\Lambda_jf\|_{\tilde\Hil_j}^2 \leq B\,\|f\|_\Hil^2,
$$
for all $f\in\Hil$.
\end{defn}

In \cite{arf} it has been proved that if
 $\Theta$ is an
alternate dual of $\Lc$, then
$\Lc$ is an alternate dual of
$\Theta$. In term of synthesis and analysis operators this means that $T_\Theta^\dagger T_\Lc=T_\Lc^\dagger T_\Theta=\1_\Hil$.

We have, \cite{na}:
\begin{lem}\label{dual}
Let $\Lc=\{\Lambda_j\in B(\Hil,\tilde\Hil_j), \,j\in J\}$ and $\Theta=\{\theta_j\in B(\Hil,\tilde\Hil_j), \,j\in J\}$ be
$g$-Bessel families of $(\Hil,\{\tilde\Hil_j\})$ such that $$f=\sum_{i\in
J}\Lambda_{i}^\dagger\theta_{i}f,\quad f\in \Hil.$$ Then $\Lc$ and $\Theta$ are $g$-frames.
\end{lem}
For two g-frames of of $(\Hil,\{\tilde\Hil_j\})$, $\Lc=\{\Lambda_j\in B(\Hil,\tilde\Hil_j), \,j\in J\}$ and $\Theta=\{\theta_j\in B(\Hil,\tilde\Hil_j), \,j\in J\}$, we can prove the following result:
\begin{prop}
Let $\Theta$ be an alternate dual $g$-frame of $\Lc$
  and let $\{e_{i,k}: k\in K_i\}$ be an orthonormal basis for $\tilde\Hil_i.$ Calling $u_{i,k}=\Lambda_{i}^\dagger e_{i,k}$ and $v_{i,k}=\theta_{i}^\dagger e_{i,k}$, $k\in K_i$, $i\in J$, then $\{v_{i,k}\}_{i\in J, k\in K_i}$ is a dual frame of
 $\{u_{i,k}\}_{i\in J, k\in K_i}.$
\end{prop}
\begin{proof}
Let $f\in \Hil$. We have
\begin{equation*}
\begin{aligned}
f=\sum_{i\in J}\theta_{i}^\dagger\Lambda_{i}f&=\sum_{i\in J}\theta_{i}^\dagger\left(\sum_{k\in K_i}\langle \Lambda_{i}f, e_{i,k}\rangle_{\tilde\Hil_i} e_{i,k} \right)=
\sum_{i\in J}\sum_{k\in K_i}\langle \Lambda_{i}f, e_{i,k}\rangle_{\tilde\Hil_i} v_{i,k}
\\
&=\sum_{i\in J}\sum_{k\in K_i}\left\langle  f, \Lambda_{i}^\dagger e_{i,k}\right\rangle_{\Hil} v_{i,k}
=\sum_{i\in J}\sum_{k\in K_i}\langle f, u_{i,k}\rangle v_{i,k}
\end{aligned}
\end{equation*}
\end{proof}

A similar result was also recently obtained in \cite{jove}. Let $\Lc=\{\Lambda_j\in B(\Hil,\tilde\Hil_j), \,j\in J\}$ and $\Theta=\{\theta_j\in B(K,\tilde\Hil_j), \,j\in J\}$ be $g$-frames of $(\Hil,\{\tilde\Hil_j\})$ respectively, where $K$ is another Hilbert space, in general different from $\Hil$. We say that $\Lc$ and $\Theta$ are similar if there is a bounded invertible operator $U:\Hil\rightarrow K$ so that
$\Lambda_{j}=\Theta_{j}U$ for all $j\in J$. In \cite{na} the following result is proved for $\Hil=K$:

\begin{prop}\label{kam}
Let $\Lc=\{\Lambda_j\in B(\Hil,\tilde\Hil_j), \,j\in J\}$ and $\Theta=\{\theta_j\in B(\Hil,\tilde\Hil_j), \,j\in J\}$ be $g$-frames of $(\Hil,\{\tilde\Hil_j\})$. $\Lc$ and $\Theta$ are similar if and only if  their analysis operators have the same ranges.
\end{prop}

\begin{prop}
Let $\Lc=\{\Lambda_j\in B(\Hil,\tilde\Hil_j), \,j\in J\}$  be a $g$-frame of $(\Hil,\{\tilde\Hil_j\})$. Then $\Theta=\{\theta_j\in B(\Hil,\tilde\Hil_j), \,j\in J\}$ is the canonical dual of $\Lambda$, $\tilde\Lc$, if and only if $\|T_{\Theta}f\|\leq\|T_{\Gamma}f\|$ for all $f\in\Hil$ and for each dual $g$-frame $\Gamma$ of $\Lc$.
\end{prop}

\begin{proof}
Suppose first that $\Theta\equiv\widetilde{\Lambda}$, the canonical dual of $\Lambda$, and $\Gamma=\{\Gamma_j\in B(\Hil,\tilde\Hil_j), \,j\in J\}$ is a (generic) dual of $\Lambda$.
Then, for all $f,g\in \Hil$,
$$
\langle T_{\Gamma}f-T_{\widetilde{\Lc}}f,T_{\Lc}g\rangle_{\hat\Hil}=\langle T_{\Gamma}f,T_{\Lc}g\rangle_{\hat\Hil}-\langle T_{\widetilde{\Lc}}f,T_{\Lc}g\rangle_{\hat\Hil}=
$$
$$
=\langle f,T_{\Gamma}^\dagger\,T_{\Lc}g\rangle_{\Hil}-\langle f,T_{\widetilde{\Lc}}^\dagger T_{\Lc}g \rangle_{\Hil}=\langle f,\,g\rangle_{\Hil}-\langle f,\,g\rangle_{\Hil}=0,
$$
since $T_{\Gamma}^\dagger\,T_{\Lc}=T_{\widetilde{\Lc}}^\dagger T_{\Lc}=\1_\Hil$. Therefore
 $Range(T_{\Gamma}-T_{\widetilde{\Lc}})\perp Range(T_{\Lc})$. Now, since $\widetilde{\Lc}$ and $\Lc$ are similar, Proposition 10 implies that
$Range(T_{\Gamma}-T_{\widetilde{\Lc}})\perp Range(T_{\widetilde{\Lc}})$ so that, for all $f,g\in\Hil$, $\langle T_{\Gamma}f-T_{\widetilde{\Lc}}f,T_{\tilde\Lc}g\rangle_{\hat\Hil}=0$. Then a direct computation shows that
\be
\|T_\Gamma f\|^2_{\hat\Hil}=\|T_\Gamma f-T_{\widetilde{\Lc}}f\|^2_{\hat\Hil}+\|T_{\widetilde{\Lc}} f\|^2_{\hat\Hil},
\label{3add1}\en
so that \be\|T_{\widetilde{\Lc}} f\|_{\hat\Hil}\leq \|T_\Gamma f\|_{\hat\Hil}\label{30}\en for all duals $\Gamma$ of $\Lc$ and for all $f\in\Hil$. Now, since $\Theta=\tilde\Lc$ by assumption, the statement follows.

\vspace{2mm}

Let us now prove the vice-versa. Hence we assume that $\|T_{\Theta} f\|_{\hat\Hil}\leq \|T_\Gamma f\|_{\hat\Hil}$ for all dual g-frame $\Gamma$ of $\Lc$ and for all $f\in\Hil$. We want to show that $\Theta=\tilde\Lc$.

First, since $\tilde\Lc$ is a dual g-frame of $\Lc$, $\|T_{\Theta} f\|_{\hat\Hil}\leq \|T_{\tilde\Lc} f\|_{\hat\Hil}$, $\forall f\in\Hil$.  On the other hand, using (\ref{30}) with $\Gamma=\Theta$, we find  $\|T_{\tilde\Lc} f\|_{\hat\Hil}\leq \|T_{\Theta} f\|_{\hat\Hil}$, $\forall f\in\Hil$. Hence, for all $f\in\Hil$, $\|T_{\tilde\Lc} f\|_{\hat\Hil}= \|T_{\Theta} f\|_{\hat\Hil}$. The same procedure which produces (\ref{3add1}) also shows that
$
\|T_\Theta f\|^2_{\hat\Hil}=\|T_\Theta f-T_{\widetilde{\Lc}}f\|^2_{\hat\Hil}+\|T_{\widetilde{\Lc}} f\|^2_{\hat\Hil},
$
which, because of the previous equality, implies that $\|T_\Theta f-T_{\widetilde{\Lc}}f\|^2_{\hat\Hil}=0$. Hence $T_\Theta f=T_{\widetilde{\Lc}}f$ for all $f\in\Hil$ and, as a consequence, $\Theta=\tilde\Lc$.

\end{proof}

\begin{prop}\label{gr4}
If $\Lc=\{\Lambda_j\in B(\Hil,\tilde\Hil_j), \,j\in J\}$   is a $g$-Riesz basis  of $(\Hil,\{\tilde\Hil_j\})$, then there exists a unique sequence $\Theta=\{\theta_{j}\in B(\Hil,\tilde\Hil_{j}), \,j\in J\}$ such that $f=\sum_{j\in J}\theta_{j}^\dagger\Lambda_{j}f$ for each $ f\in \Hil$, and $\Theta=\tilde\Lc$.
\end{prop}
\begin{proof}
Let $\Lc$   be a $g$-Riesz basis  of $(\Hil,\{\tilde\Hil_j\})$ and $\Theta$  an alternate dual $g$-frame of $\Lc$. Then, for all $ f\in \Hil$,
\begin{equation}\label{gr1}
 \sum_{i\in J}\Lambda_{i}^\dagger(\widetilde{\Lambda}_{i}-\theta_{i})f=0.
\end{equation}
Since $\Lc$ is a $g$-Riesz basis, its synthesis operator $T_\Lc^\dagger$ is injective. Hence (\ref{gr1}) implies that $(\widetilde{\Lambda}_{i}-\theta_{i})f=0$, for all $ f\in \Hil,$ and so $\widetilde{\Lambda}_{i}=\theta_{i},$ for all $ i\in J$, and $\tilde\Lc=\Theta$ as a consequence.

\end{proof}

\begin{prop}\label{gr2}
Let $\Lc=\{\Lambda_j\in B(\Hil,\tilde\Hil_j), \,j\in J\}$   be a $g$-frame of  $(\Hil,\{\tilde\Hil_j\})$. Then $\Lc$ is a $g$-Riesz basis if and only if $Range (T_{\Lc})=\hat\Hil$.
\end{prop}
\begin{proof}
Let us first prove that if $\Lc$ is a $g$-Riesz basis then $Range (T_{\Lc})=\hat\Hil$.

Indeed we know that there exists a bounded invertible operator $X\in B(\Hil)$ and a g-on basis $\Theta=\{\theta_{j}\in B(\Hil,\tilde\Hil_{j}), \,j\in J\}$ such that $\Lambda_j=\theta_j X$ for all $j\in J$. Hence, taking $f\in\Hil$, $T_\Lc f=T_\Theta (Xf)$. Let now $\underline h\in\hat\Hil$ be in $(Range(T_\Lc))^\perp$. We will show that $\underline h=0$. Infact, for $f\in\Hil$ we have
$$
0=\left<\underline h,T_\Lc f\right>_{\hat\Hil}=\left<T_\Theta^\dagger\,\underline h, X\,f\right>_{\Hil}
$$
which implies, because of the arbitrariness of $X\,f$ in $\Hil$, that  $T_\Theta^\dagger\,\underline h=0$. Now, since $T_\Theta T_\Theta^\dagger=\1_{\hat\Hil}$, we deduce that  $\underline h=0$. Hence $Range (T_{\Lc})=\hat\Hil$.

\vspace{1mm}

Let now $Range (T_{\Lc})=\hat\Hil$. Then, since $\ker(T_\Lc^\dagger)=(Range(T_\Lc))^\perp=\{0\}$, $T_\Lc^\dagger$ is injective. Moreover, $T_\Lc^\dagger$ is surjective since any $f\in\Hil$ can be written as $f=T_\Lc^\dagger\left(T_{\tilde\Lc}f\right)$. Hence $T_\Lc^\dagger$ is invertible. Then it is  simple to check that, for all $\underline g\in \hat\Hil$,
$$
\|(T_\Lc^\dagger)^{-1}\|^{-2}\sum_{j\in J}\|g_j\|^2_{\tilde\Hil_j}\leq  \left\|\sum_{j\in J}\Lambda_j^\dagger g_j\right\|^2_\Hil\leq \|T_\Lc^\dagger\|^{2}\sum_{j\in J}\|g_j\|^2_{\tilde\Hil_j}.
$$
Moreover it is also clear that $\Lc$ is g-complete. Hence, \cite{sun}, $\Lc$ is a g-Riesz basis.

\end{proof}

A similar result was proved in \cite{zhu}, while in \cite{licao} the authors prove a result close to the following theorem.

\begin{thm}\label{gr3}
Let $\Lc=\{\Lambda_j\in B(\Hil,\tilde\Hil_j), \,j\in J\}$   be a $g$-frame of  $(\Hil,\{\tilde\Hil_j\})$ but not a $g$-Riesz basis. Then $\Lc$ has a dual $g$-frame which is different from its
canonical dual $\tilde\Lc$.
\end{thm}
\begin{proof}
Let  $\Lc$ be such a $g$-frame. Then by Proposition \ref{gr2}, $Range (T_{\Lc})\neq\hat\Hil$. Hence there exists a non zero ${\underline F}\in (Range (T_{\Lc}))^{\perp}$. Of course we can always assume that $\|{\underline F}\|_{\hat\Hil}=1.$ We use ${\underline F}$ to define a family of bounded operators
$Q_j:\hat\Hil\rightarrow\tilde\Hil_j$ as follows:
$$ Q_j({\underline G})=\langle {\underline F},{\underline G}\rangle_{\hat\Hil} F_j\in \tilde\Hil_j,$$ for all
$j\in J.$ The sequence $\{Q_j\in B(\hat\Hil,\tilde\Hil_j),\,j\in J\}$ is a $g$-Bessel family of  $(\hat\Hil,\{\tilde\Hil_j\})$. Indeed, recalling that $\|{\underline F}\|_{\hat\Hil}^2=\sum_{i\in J}\|F_i\|_{\tilde\Hil_i}=1$, we have
$$\sum_{i\in J}\|Q_i({\underline G})\|^2=\sum_{i\in J}|\langle {\underline F},{\underline G}\rangle_{\hat\Hil} |^2\|F_i\|_{\tilde\Hil_i}^2\leq\|{\underline G}\|_{\hat\Hil}\,,\quad {\underline G}\in \hat\Hil.$$
Let now $U:\Hil\rightarrow \hat\Hil$ be a bounded invertible linear operator. Then
$\{Q_jU\in B(\Hil,\tilde\Hil_j),\,j\in J\}$ is a $g$-Bessel family of $(\Hil,\{\tilde\Hil_j\})$.

Let $f,g\in\Hil.$ Since ${\underline F}\in (Range (T_{\Lc}))^{\perp}$, we find that
$$
\sum_{i\in J}\langle Q_iUf,\Lambda_i g\rangle_{\tilde\Hil_i}=\sum_{i\in J}\left\langle\langle {\underline F},Uf\rangle _{\hat\Hil} F_i,\Lambda_i g \right\rangle_{\tilde\Hil_i}=
\langle Uf, {\underline F}\rangle_{\hat\Hil}\sum_{i\in J} \langle F_i,\Lambda_i g\rangle_{\tilde\Hil_i}=
$$
$$
=\langle Uf, {\underline F}\rangle_{\hat\Hil} \langle {\underline F}, T_\Lc g\rangle_{\hat\Hil}=0.
$$
 Calling as usual $\tilde\Lc=\{\widetilde{\Lambda}_{j}\in B(\Hil,\tilde\Hil_j), \,j\in J\}$ the canonical dual of $\Lc$, this implies that, for $f,g\in\Hil$,
$$
\sum_{i\in J}\langle (\widetilde{\Lambda}_{i}+Q_iU)f,\Lambda_i g\rangle_{\tilde\Hil_i}=\sum_{i\in J}\langle \widetilde{\Lambda}_{i}f ,\Lambda_i g \rangle_{\tilde\Hil_i}=\sum_{i\in J}\langle \Lambda^\dagger_{i}\widetilde{\Lambda}_{i}f , g \rangle_\Hil=\langle f,g \rangle_\Hil.
$$
Therefore the set $\Gamma:=\{\widetilde{\Lambda}_{j}+Q_jU\in B(\Hil,\tilde\Hil_j), \,j\in J\}$  is a  dual of $\Lc$.
To check that $\Gamma\neq\tilde\Lc$ it is enough to see that, taking $h=U^{-1}{\underline F}$, $(Q_iU)(h)=Q_i(\underline F)=\langle \underline F, {\underline F}\rangle_{\hat\Hil}F_i=F_i$. Since $\underline F\neq 0$, $F_i$ is different from 0 for some $i\in J$. Hence $Q_iU$ is not identically zero.

\end{proof}

\begin{thm}
Let $\Theta=\{\theta_{i}\in B(\Hil,\tilde\Hil_{j}),\,j\in J\}$ be a $g$-frame of $(\Hil,\{\tilde\Hil_j\})$, which is dual of a $g$-frame $\Lc=\{\Lambda_j\in B(\Hil,\tilde\Hil_j), \,j\in J\}$. $\Theta$ is the canonical dual of $\Lc$, $\tilde\Lc$, if
and only if $T_{\Theta}^\dagger T_{\Theta}=T^\dagger_{\Theta}T_{\Gamma}$ for all duals $\Gamma$ of $\Lc$.
\end{thm}

\begin{proof}
Let us first assume that $\Theta=\tilde\Lc$, and let $\Gamma$ be any dual for $\Lc$. Then $\langle (T_\Theta-T_\Gamma)f, T_\Lc g\rangle_{\hat\Hil}= \langle f,  g\rangle_{\Hil}- \langle f,  g\rangle_{\Hil}=0$. Since $\Theta(=\tilde\Lc)$ and $\Lc$ are similar then $Range(T_\Lc)=Range(T_\Theta)$. Therefore $\langle (T_\Theta-T_\Gamma)f, T_\Theta h\rangle_{\hat\Hil}= 0$ for all $h\in\Hil$, and $\langle T_\Theta^\dagger(T_\Theta-T_\Gamma)f,  h\rangle_{\Hil}= 0$, for all $f,g\in\Hil$. Therefore $T_{\Theta}^\dagger T_{\Theta}=T^\dagger_{\Theta}T_{\Gamma}$.

Viceversa, let $T_{\Theta}^\dagger T_{\Theta}=T^\dagger_{\Theta}T_{\Gamma}$ for all duals $\Gamma$ of $\Lc$. First we observe that if $\Lc$ is a g-Riesz basis, then its dual is unique. Hence $\Theta=\tilde\Lc$.

Suppose rather that $\Lc$ is not a g-Riesz basis. The, by Proposition 13, $Range (T_{\Lc})\subset\hat\Hil$. Therefore we can proceed as in the proof of Theorem 14: we introduce a normalized vector $\underline F\in\hat\Hil$, $\underline F\in(Range (T_{\Lc}))^\perp$. Then we define the same operators $Q_j$ and a g-Bessel family $\Gamma:=\{\theta_{j}+Q_jU\in B(\Hil,\tilde\Hil_j), \,j\in J\}$, where $U$ is as before. Since $\Theta$ is a dual g-frame of $\Lc$, $\Gamma$ is also a dual g-frame of $\Lc$. Then $(T_\Gamma-T_\Theta)f=\{Q_jUf\}_{j\in J}$ for all $f\in\Hil$. If we take in particular $f=U^{-1}{\underline F}$, then $\{Q_jUf\}_{j\in J}={\underline F}$. Moreover, using our assumption and the previous equality, we find $T_\Theta^\dagger(\underline F)=0$, so that $\underline F\in \ker(T_\Theta^\dagger)=(Range (T_{\Theta}))^\perp$.
Hence $(Range (T_{\Lc}))^\perp\subseteq (Range (T_{\Theta}))^\perp$. Suppose now that the inclusion is strict. Hence we have a non zero $\underline f\in Range (T_{\Theta}))^\perp$ which does not belong to $Range (T_{\Lc}))^\perp$. Then $\underline f=T_\Lc h$ for some $h\in\Hil$ and $\left<\underline f, T_\Theta g\right>_{\hat\Hil}=0$ for all $g\in\Hil$. If in particular we take $g=h$ then we conclude that $0=\left<\underline f, T_\Theta h\right>_{\hat\Hil}=\|h\|^2_\Hil$, so that $h=0$ and $\underline f=0$, which is against the original assumption. Hence $(Range (T_{\Lc}))^\perp = (Range (T_{\Theta}))^\perp$ and, as a consequence of Proposition 10, $\Lc$ and $\Theta$ are similar. Therefore $\Lambda_j=\theta_j X$ for a certain invertible operator $X\in B(\Hil)$. But also $\Lc$ and $\tilde \Lc$ are similar since $\tilde\Lambda_j=\Lambda_j S_\Lc^{-1}$.
Furthermore, since  $\Theta$ and $\tilde\Lc$ are both duals of $\Lc$, we deduce that $\left<T_\Theta f-T_{\tilde\Lc}f,T_\Lc g\right>_{\hat\Hil}=0$ for all $f$ and $g$ in $\Hil$. Then, if we take $g=(X-S_\Lc^{-1})f$, we find that $0=\left<T_\Theta f-T_{\tilde\Lc}f,T_\Lc (X-S_\Lc^{-1})f\right>_{\hat\Hil}=\sum_{i\in J}\|\Lambda_i(X-S_\Lc^{-1})f\|_{\tilde\Hil_i}^2$, so that $\Lambda_i(X-S_\Lc^{-1})f=0$ for all $i\in J$ and, consequently, that $\sum_{i\in J}\tilde \Lambda_i^\dagger\Lambda_i(X-S_\Lc^{-1})f=(X-S_\Lc^{-1})f=0$, for all $f\in\Hil$. Hence $X=S_\Lc^{-1}$, and $\tilde\Lc=\Theta$.

\end{proof}

\begin{prop}
let $\Lc=\{\Lambda_j\in B(\Hil,\tilde\Hil_j), \,j\in J\}$ be a $g$-frame of $(\Hil,\{\tilde\Hil_j\})$ and let $\Theta=\{\theta_{j}\in B(\Hil,\tilde\Hil_{j}),\,j\in J\}$ be a sequence of bounded operators. The following
are equivalent:
\begin{enumerate}
\item[(1)] $\Theta$ is a $g$-frame of $(\Hil,\{\tilde\Hil_j\})$.
\item[(2)] There is a constant $M>0$ so that, for all $f\in\Hil$, we have
\be\sum_{i\in J}\|\Lambda_{i} f-\theta_{i}f\|_{\tilde\Hil_i}^2\leq M\: \min \left(\sum_{i\in J}\|\Lambda_{i} f\|_{\tilde\Hil_i}^2,\sum_{i\in J}\|\theta_{i} f\|_{\tilde\Hil_i}^2\right).\label{3add2}\en
\end{enumerate}
Moreover, (2) implies (3) below and, if $\Theta$ is a $g$-Bessel family and (3) holds, then $\Theta$ is also a g-frame.
\begin{enumerate}
\item[(3)] There is a constant $M>0$ such that for all $f\in\Hil$ we have
$$\sum_{i\in J}\|\Lambda_{i} f-\theta_{i}f\|_{\tilde\Hil_i}^2\leq M\sum_{i\in J}\|\theta_{i} f\|_{\tilde\Hil_i}^2$$
\end{enumerate}
\end{prop}
\begin{proof}
$(1)\Rightarrow (2)$: Let $A$ and $B$ be the $g$-frame bounds for $\Lc$, and $C$ and $D$ be the $g$-frame bounds for $\Theta.$ Then, for all $f\in\Hil$
$$
\sum_{i\in J}\|\Lambda_{i} f-\theta_{i}f\|_{\tilde\Hil_i}^2 \leq2\sum_{i\in J}\|\Lambda_{i} f\|_{\tilde\Hil_i}^2+2\sum_{i\in J}\|\theta_{i}f\|_{\tilde\Hil_i}^2
\leq 2\sum_{i\in J}\|\Lambda_{i} f\|_{\tilde\Hil_i}^2+2D\|f\|_\Hil^2$$
$$\leq 2\sum_{i\in J}\|\Lambda_{i} f\|_{\tilde\Hil_i}^2+2\frac{D}{A}\sum_{i\in J}\|\Lambda_{i} f\|_{\tilde\Hil_i}^2
\leq 2\left(1+\frac{D}{A}\right)\sum_{i\in J}\|\Lambda_{i} f\|_{\tilde\Hil_i}^2.
$$
Since $\Theta$ is $g$-frame by the same argument we have
$$\sum_{i\in J}\|\Lambda_{i} f-\theta_{i}f\|_{\tilde\Hil_i}^2\leq 2\left(1+\frac{B}{C}\right)\sum_{i\in J}\|\theta_{i} f\|_{\tilde\Hil_i}^2.$$
Hence (\ref{3add2}) follows.

$(2)\Rightarrow (1):$ For $M$ given in (2) and any $f\in\Hil$ we have, recalling that $\Lc$ is a g-frame with bounds $A$ and $B$,
$$
A\|f\|_{\Hil}^2 \leq\sum_{i\in J}\|\Lambda_{i} f\|_{\tilde\Hil_i}^2
\leq 2\sum_{i\in J}\|\Lambda_{i} f-\theta_{i}f\|_{\tilde\Hil_i}^2+2\sum_{i\in J}\|\theta_{i}f\|_{\tilde\Hil_i}^2
\leq 2\left(M\sum_{i\in J}\|\theta_{i} f\|_{\tilde\Hil_i}^2+\sum_{i\in J}\|\theta_{i}f\|_{\tilde\Hil_i}^2\right)$$
$$=2(M+1)\sum_{i\in J}\|\theta_{i}f\|_{\tilde\Hil_i}^2
\leq 4(M+1)\left(\sum_{i\in J}\|\Lambda_{i} f-\theta_{i}f\|_{\tilde\Hil_i}^2+\sum_{i\in J}\|\Lambda_{i}f\|_{\tilde\Hil_i}^2\right)$$
$$
\leq 4(M+1)\left(M\sum_{i\in J}\|\Lambda_{i}f\|_{\tilde\Hil_i}^2+\sum_{i\in J}\|\Lambda_{i}f\|_{\tilde\Hil_i}^2\right)
\leq 4(M+1)^2\sum_{i\in J}\|\Lambda_{i}f\|_{\tilde\Hil_i}^2
\leq 4(M+1)^2B\|f\|_{\Hil}^2.
$$
Then it follows that $$\frac{A}{2M+2}\|f\|_{\Hil}^2\leq\sum_{i\in J}\|\theta_{i}f\|_{\tilde\Hil_i}^2\leq 2B(M+1)\|f\|_{\Hil}^2,$$
which means that $\Theta$ is a g-frame of $(\Hil,\{\tilde\Hil_j\})$.

It is clear that (2) implies (3).

Assume now that $\Theta$ is a $g$-Bessel family and that (3) holds. Then $\Theta$ has lower $g$-frame bound. Indeed we find, taking
$f\in\Hil$, and with similar estimates as those in the proof of implication $(2)\Rightarrow (1)$, $$\frac{A}{2M+2}\|f\|_{\Hil}^2\leq\sum_{i\in J}\|\theta_{i}f\|_{\tilde\Hil_i}^2.$$
Hence $\Theta$ is a g-frame of $(\Hil,\{\tilde\Hil_j\})$.
\end{proof}

The equivalence (1) $\Leftrightarrow$ (2) is also proved in  \cite{yao}. In two next propositions, we generalized the results of Gavruta \cite{ga} for fusion frames to $g$-frames.
\begin{prop}\label{g1}
Assume that $\Lc=\{\Lambda_j\in B(\Hil,\tilde\Hil_j), \,j\in J\}$  and $\Theta=\{\theta_{j}\in B(\Hil,\tilde\Hil_{j}),\,j\in J\}$ are $g$-Bessel families of $(\Hil,\{\tilde\Hil_j\})$ with $g$-Bessel bounds $B_1,B_2$ respectively, and that there exist $m<1,n>-1$
such that $$\|f-\sum_{i\in J}\Lambda_{i}^\dagger\theta_{i} f\|_\Hil\leq m\|f\|_\Hil+n\|\sum_{i\in J}\Lambda_{i}^\dagger\theta_{i}f\|_\Hil,\: f\in\Hil$$ then $\Theta$
is a $g$-frame for $\Hil$.
\end{prop}
\begin{proof}
Let us first define the operator $V:\Hil\rightarrow\Hil$ by $Vf=\sum_{i\in J}\Lambda_{i}^\dagger\theta_{i} f.$ We have
\begin{equation}\label{b2}
\|Vf\|_\Hil=sup_{\|g\|\leq 1} |\langle V\,f,g\rangle_\Hil| \leq \sup_{\,\|g\|\leq 1} \left|\langle \sum_{i\in J}\Lambda_{i}^\dagger\theta_{i} f,g\rangle_\Hil\right|
\end{equation}
\begin{equation*}
\begin{aligned}
&\leq \sup_{\|g\|\leq 1}\left(\sum_{i\in J}\|\Lambda_{i}g\|_{\tilde\Hil_i}^2 \right)^\frac{1}{2}\left(\sum_{i\in J}\|\theta_{i}f\|_{\tilde\Hil_i}^2 \right)^\frac{1}{2}
\\
&\leq \sqrt{B_1}\left(\sum_{i\in J}\|\theta_{i}f\|_{\tilde\Hil_i}^2 \right)^\frac{1}{2}
\leq \sqrt{B_1}\sqrt{B_2}\|f\|_\Hil.
\end{aligned}
\end{equation*}
It follows that $V$ is well defined and bounded by $\sqrt{B_1B_2}$. On the other hand, we have
$$\|f\|_\Hil-\|Vf\|_\Hil\leq\|f-Vf\|_\Hil\leq m\|f\|_\Hil+n\|Vf\|_\Hil,\quad f\in\Hil,$$ and so $$\|Vf\|_\Hil\geq \frac{1-m}{1+n}\|f\|_\Hil.$$
Hence, using this and inequality (\ref{b2}), we obtain $$\sum_{i\in J}\|\theta_{i}f\|_{\tilde\Hil_i}^2 \geq\frac{1}{B_1}\left(\frac{1-m}{1+n}\right)^2\|f\|_\Hil^2.$$
\end{proof}

\begin{prop}
Let $\Lc=\{\Lambda_j\in B(\Hil,\tilde\Hil_j), \,j\in J\}$  and $\Theta=\{\theta_{j}\in B(\Hil,\tilde\Hil_{j}),\,j\in J\}$ be two $g$-Bessel families of $(\Hil,\{\tilde\Hil_j\})$ with $g$-Bessel bounds $B_1$ and $B_2$ respectively. Suppose that there exists $0\leq m<1,$
such that $$\|f-\sum_{i\in J}\Lambda_{i}^\dagger\theta_{i} f\|_\Hil\leq m\|f\|_\Hil,\: f\in\Hil.$$ Then $\Theta$ and
$\Lc$ are $g$-frames.
\end{prop}
\begin{proof}
In Proposition \ref{g1}, let us consider $n=0$. Then $\Theta$ is a $g$-frame with lower bound
$\frac{1}{B_1}(1-m)^2.$ To prove that $\Lc$ is also a g-frame,
we define the operator $W:\Hil\rightarrow\Hil$ by $Wf=\sum_{i\in J}\theta_{i}^\dagger\Lambda_{i} f=V^\dagger f$.
Then $W$ is well defined and, as in the proof of the previous proposition,
\be\label{b3}
\|Wf\|_\Hil=\sup_{\|g\|_\Hil\leq 1} |\langle Wf,g\rangle_\Hil| \leq \sup_{\|g\|_\Hil\leq 1} \left|\langle \sum_{i\in J}\theta_{i}^\dagger\Lambda_{i} f,g\rangle_\Hil\right|
\leq \sqrt{B_2}\left(\sum_{i\in J}\|\Lambda_{i}f\|_{\tilde\Hil_i}^2 \right)^\frac{1}{2}.
\en

Then, for all $f\in\Hil$, recalling that $W=V^\dagger$,
$$
\|f\|_\Hil-\|Wf\|_\Hil\leq\|f-Wf\|_\Hil= \|(I-V)^\dagger f\|_\Hil
\leq m\|f\|_\Hil.
$$
Hence $\|Wf\|_\Hil\geq (m-1)\|f\|_\Hil,$ for all $f\in\Hil$. Now (\ref{b3}) implies that
$$\sum_{i\in J}\|\Lambda_{i}f\|_{\tilde\Hil_i}^2 \geq\frac{1}{B_2}(m-1)^2\|f\|_\Hil^2.$$
\end{proof}

\section{Coherent states and g-sets}

Since the birth of quantum mechanics coherent states have always had a very central role: they are, in fact, the quantum states which are closer to classical ones. This is reflected from the fact that they minimize the well-known Heisenberg uncertainty relation, $\Delta x\,\Delta p =\frac{\hbar}{2}$, where $x$ and $p$ are the position and the momentum operators of a quantum particle. Coherent states are often defined as eigenstates of a certain {\em lowering operator} related to $x$ and $p$, $a=\frac{1}{\sqrt{2}}(a+ip)$, and are useful also because they produce a certain  resolution of the identity, which is a continuous version of (\ref{29}). Many details on coherent states can be found in \cite{ks} or, for more updated references, in \cite{aag} or in \cite{gazeaubook}. It is easy to check, looking at the literature, that there is not an unique way to introduce these states. On the contrary, quite often different authors call {\em coherent} different vectors of a certain Hilbert space, depending on what they are interested in.

We will discuss here how the notion of coherent states can be exported to the present {\em g-setting}, with no particular difficulty. It is not hard to imagine that the procedure we will discuss in this section is not unique, and other interesting alternates may exist.

In what follows we will take $J=K_j={\Bbb N}_0:={\Bbb N}\cup\{0\}$, for all $j\in J$, and $\Theta=\{\theta_j\in B(\Hil,\{\tilde\Hil_j\}), \,j\in {\Bbb N}_0\}$ will be a g-on basis of $(\Hil,\{\tilde\Hil_j\})$. In view of what discussed in Section II this means that, taken a fixed o.n. basis $\E_j=\{e_k^{(j)},\, k\in {\Bbb N}_0\}$ in $\tilde\Hil_j$, the set $\T$ defined as  in (\ref{215}), $\T:=\{t_k^{(j)}:=\theta_j^\dagger\,e_k^{(j)},\, k, j\in {\Bbb N}_0\}$, is an o.n. basis of $\Hil$. This implies that $\theta_j f=\sum_{k=0}^\infty\left<t_k^{(j)},f\right>_\Hil\,e_k^{(j)}$ for all $f\in\Hil$, and that $\theta_j^\dagger \tilde f=\sum_{k=0}^\infty\left<e_k^{(j)},\tilde f\right>_{\tilde\Hil_j}\,t_k^{(j)}$ for all $\tilde f\in\tilde\Hil_j$. Using the o.n. basis $\T$ we can introduce the following vectors of $\Hil$
\be
\Phi_\Theta(z,w):=N(z,w)\,\sum_{k,l=0}^\infty\,\frac{z^k\,w^l}{\sqrt{k!\,l!}}\,\,t_k^{(l)},
\label{51}\en
where $z,w\in\Bbb{C}$ and $N(z,w)$ is a normalization constant. These are a two-dimensional version of standard coherent states. More in particular, it is easy to see that $\Phi_\Theta(z,w)$ is normalized in $\Hil$ if we take $N(z,w)=e^{-(|z|^2+|w|^2)/2}$, and that the double series in (\ref{51}) converge for all $z$ and $w$ in $\Bbb{C}$. Since $t_k^{(j)}:=\theta_j^\dagger\,e_k^{(j)}$ we can rewrite (\ref{51})  as
\be
\Phi_\Theta(z,w)=e^{-|w|^2/2}\,\sum_{l=0}^\infty\,\frac{w^l}{\sqrt{l!}}\,\theta_l^\dagger \chi_l(z), \quad \mbox{where} \quad
\chi_l(z)=e^{-|z|^2/2}\,\sum_{k=0}^\infty\,\frac{z^k}{\sqrt{k!}}\,e_k^{(l)}.
\label{52}\en
$\chi_l(z)$ is normalized  in $\tilde\Hil_l$: $\left<\chi_l(z),\chi_l(z)\right>_{\tilde\Hil_l}=1$ for all $z\in {\Bbb C}$. Actually, $\chi_l(z)$ is a standard coherent state in $\tilde\Hil_l$, \cite{ks}.
Let us further define two operators $a$ and $b$ on $\Hil$ via their action on the o.n. basis $\T$. We define
\be \left\{\begin{array}{lll}
a\, t_k^{(l)}=\sqrt{k}\,t_{k-1}^{(l)},\qquad \forall l\in {\Bbb N}_0,\\
b\, t_k^{(l)}=\sqrt{l}\,t_{k}^{(l-1)},\qquad \forall k\in {\Bbb N}_0,\\
\end{array}\right.\label{53}\en
with the usual understanding that $a\, t_0^{(l)}=b\, t_k^{(0)}=0$, for all $k$ and $l$. The operators $a$ and $b$ commute in the following sense: $[a,b]t_k^{(l)}=0$ for all $k$ and $l$. Their adjoints can be computed easily and we get
$a^\dagger\, t_k^{(l)}=\sqrt{k+1}\,t_{k+1}^{(l)}$ and $b^\dagger\, t_k^{(l)}=\sqrt{l+1}\,t_{k}^{(l+1)}$. The vector $\Phi_\Theta(z,w)$ is an eigenstate of both $a$ and $b$, with eigenvalues $z$ and $w$ respectively. Indeed we can check that
\be
a\, \Phi_\Theta(z,w)=z\, \Phi_\Theta(z,w) \quad \mbox{and}\quad
b\, \Phi_\Theta(z,w)=w\, \Phi_\Theta(z,w),
\label{54}\en
for all $z, w\in \Bbb{C}$.
Another feature of these states is that they satisfy a resolution of the identity in $\Hil$. Indeed we get
\be
\frac{1}{\pi^2}\int_{\Bbb{C}}dz\int_{\Bbb{C}}dw |\Phi_\Theta(z,w)\left>\right<\Phi_\Theta(z,w)|=\sum_{k,l=0}^\infty |t_k^{(l)}\left>\right<t_k^{(l)}|=\1_\Hil.
\label{55}\en
Let us now define the following operators
$$
q_a:=\frac{a+a^\dagger}{\sqrt{2}}, \quad p_a:=\frac{a-a^\dagger}{\sqrt{2}\,i}, \quad q_b:=\frac{b+b^\dagger}{\sqrt{2}}, \quad p_b:=\frac{b-b^\dagger}{\sqrt{2}\,i}.
$$
Then, if for a generic operator $X$ we define the quantity  $$\Delta X=\sqrt{\left< \Phi_\Theta(z,w),X^2\, \Phi_\Theta(z,w)\right>-\left< \Phi_\Theta(z,w),X\, \Phi_\Theta(z,w)\right>^2},$$ we get
$$
\Delta q_a\,\Delta p_a=\frac{1}{2}, \qquad \Delta q_b\,\Delta p_b=\frac{1}{2},
$$
so that the Heisenberg uncertainty relation is saturated. For all these reasons we say that the set $\C_\Theta:=\{\Phi_\Theta(z,w), \,z,w\in\Bbb{C}\}$ is {\em a family of 2-dimensional coherent states associated to a g-on basis}.

\vspace{2mm}

It is interesting to discuss what happens if our original set $\Theta$ is replaced by a set $\Lc=\{\Lambda_j\in B(\Hil,\tilde\Hil_j), \,j\in {\Bbb N}_0\}$ which is a g-Riesz rather than a g-on basis of $(\Hil,\{\tilde\Hil_j\})$. If this is the case, then,  there exists a bounded operator $X\in B(\Hil)$ with bounded inverse $X^{-1}\in B(\Hil)$ and a g-on basis which we again call $\Theta=\{\theta_j\in B(\Hil,\tilde\Hil_j), \,j\in {\Bbb N}_0\}$ of $(\Hil,\{\tilde\Hil_j\})$, such that
$\Lambda_j=\theta_j\,X$ for all $j=0,1,2,\ldots$. Hence, it is natural to define a second set of vectors extending formulas (\ref{51}) and  (\ref{52}):
\be
\Phi_\Lc(z,w)=e^{-(|z|^2+|w|^2)/2}\,\sum_{k,l=0}^\infty\,\frac{z^k\,w^l}{\sqrt{k!\,l!}}\,u_k^{(l)}=e^{-|w|^2/2}\,\sum_{l=0}^\infty\,\frac{w^l}{\sqrt{l!}}\,\Lambda_l^\dagger \chi_l(z),
\label{56}\en
where $\chi_l(z)$ is again  defined in (\ref{52}) and $\U:=\{u_k^{(l)}=\Lambda_l^\dagger e_k^{(l)}, \,k,l\in {\Bbb N}_0\}$ is a Riesz basis in $\Hil$. This is because $u_k^{(l)}=X^\dagger\,t_k^{(l)}$, since $\T$ is an o.n. basis of $\Hil$ and $X^\dagger$ is bounded with bounded inverse.

From what we have discussed in Section II we can associate to $\Lc$ a dual set $\tilde\Lc=\{\tilde\Lambda_j=\Lambda_j\,S_\Lambda^{-1}\in B(\Hil,\tilde\Hil_j), \,j\in {\Bbb N}_0\}$, where $S_\Lc=\sum_{n=0}^\infty\,\Lambda_n^\dagger\,\Lambda_n=X^\dagger\,X$, and a set of dual states of the $\Phi_\Lc$'s, as
\be
\Phi_{\tilde\Lc}(z,w)=e^{-(|z|^2+|w|^2)/2}\,\sum_{k,l=0}^\infty\,\frac{z^k\,w^l}{\sqrt{k!\,l!}}\,v_k^{(l)}=
e^{-|w|^2/2}\,\sum_{l=0}^\infty\,\frac{w^l}{\sqrt{l!}}\,\tilde\Lambda_l^\dagger \chi_l(z),
\label{57}\en
where $\V:=\{v_k^{(l)}=\tilde\Lambda_l^\dagger e_k^{(l)}=S_\Lc^{-1}\,X^\dagger\,t_k^{(l)}, \,k,l\in{\Bbb N}_0\}$ is again a Riesz basis since $S_\Lc^{-1}\,X^\dagger$ is bounded with bounded inverse.

There exists still another set of vectors in $\Hil$ which could be naturally introduced in the game, the set $\Pc:=\{p_k^{(l)}=X^{-1}\,t_k^{(l)}, \,k,l\in{\Bbb N}_0\}$, which is still a Riesz basis for the same reason. Associated to $\Pc$ we have an a-priori different set of operators $\Lc^\uparrow=\{\Lambda_j^\uparrow=\Theta_j\,(X^{-1})^\dagger\in B(\Hil,\tilde\Hil_j), \,j\in {\Bbb N}_0\}$, which is a g-Riesz basis, and we can construct the vectors
\be
\Phi_{\Lc^\uparrow}(z,w)=e^{-(|z|^2+|w|^2)/2}\,\sum_{k,l=0}^\infty\,\frac{z^k\,w^l}{\sqrt{k!\,l!}}\,p_k^{(l)}=
e^{-|w|^2/2}\,\sum_{l=0}^\infty\,\frac{w^l}{\sqrt{l!}}\,X^{-1}\theta_l^\dagger \chi_l(z).
\label{58}\en
Standard computations show that
$$
\left<\Phi_{\Lc}(z,w),\Phi_{\Lc^\uparrow}(z,w)\right>_\Hil=1,
$$
for all $z,w\in \Bbb{C}$, and that
$$
\frac{1}{\pi^2}\int_{\Bbb{C}}\,dz\int_{\Bbb{C}}\,dw\,|\Phi_\Lc(z,w)\left>\right<\Phi_{\tilde\Lc}(z,w)|=
\frac{1}{\pi^2}\int_{\Bbb{C}}\,dz\int_{\Bbb{C}}\,dw\,|\Phi_{\tilde\Lc}(z,w)\left>\right<\Phi_{\Lc}(z,w)|=\1_\Hil.
$$
Analogously we find
$$
\frac{1}{\pi^2}\int_{\Bbb{C}}\,dz\int_{\Bbb{C}}\,dw\,|\Phi_\Lc(z,w)\left>\right<\Phi_{\Lc^\uparrow}(z,w)|=
\frac{1}{\pi^2}\int_{\Bbb{C}}\,dz\int_{\Bbb{C}}\,dw\,|\Phi_{\Lc^\uparrow}(z,w)\left>\right<\Phi_{\Lc}(z,w)|=\1_\Hil.
$$

Let us now  define, starting from the operator $a$ introduced in (\ref{53}), the following operators: $a_\Lc:=X^\dagger a (X^\dagger)^{-1}$, $a_{\tilde\Lc}:=S_\Lambda^{-1} X^\dagger a (X^\dagger)^{-1} S_\Lambda$ and $a_{\Lc^\uparrow}:=X^{-1} a X$. Needless to say, $a_\Lc=a_{\Lc^\uparrow}$ if $X$ is unitary. It is now straightforward to check that
$$
a_\Lc u_k^{(l)}=\sqrt{k}\,u_{k-1}^{(l)}, \quad a_\Lc \Phi_{\Lc}(z,w)=z\Phi_{\Lc}(z,w), \quad
a_{\tilde\Lc} v_k^{(l)}=\sqrt{k}\,v_{k-1}^{(l)}, \quad  a_{\tilde\Lc} \Phi_{\tilde\Lc}(z,w)=z\Phi_{\tilde\Lc}(z,w), \quad
$$
and
$$
a_{\Lc^\uparrow} v_k^{(l)}=\sqrt{k}\,v_{k-1}^{(l)}, \quad  a_{\Lc^\uparrow} \Phi_{\Lc^\uparrow}(z,w)=z\Phi_{\Lc^\uparrow}(z,w).
$$
This means that we can define, starting from $a$, a set of different operators acting as lowering operators on the different Riesz bases considered here.
Analogous definitions and conclusions can be deduced starting with $b$ rather than with $a$. It should be mentioned that $a_\Lc$, $a_{\tilde\Lc}$ and $a_{\Lc^\uparrow}$ do not satisfy, in general,  the same canonical commutation relation as $a$ does: for instance, $[a_\Lc,a_\Lc^\dagger]\neq\1$, in general. Hence $a_\Lc^\dagger$ is not a raising operator for $\U$.

 Nevertheless, in view of the above results, we can interpret $(\Phi_\Lc,\Phi_{\Lc^\uparrow})$ as {\em bi-coherent states}. As a matter of fact, it is possible to check that $\Phi_{\Lc^\uparrow}(z,w)=\Phi_{\tilde\Lc}(z,w)$ for all $z$ and $w$. This is somehow expected but not completely trivial, from our point of view. Indeed,
since $v_k^{(l)}=S_\Lc^{-1}\,X^\dagger\,\theta_l^\dagger\,e_k^{(l)}$ and $p_k^{(l)}=X^{-1}\,\theta_l^\dagger\,e_k^{(l)}$, for all $k,l$,  a simple manipulation shows that $v_k^{(l)}=p_k^{(l)}$. Hence $\V=\Pc$ and, as a consequence, $\Phi_{\Lc^\uparrow}(z,w)=\Phi_{\tilde\Lc}(z,w)$. Moreover, since $\tilde\Lambda_l=\theta_l X S_\Lc^{-1}$ and $\Lambda_l^\uparrow=\theta_l (X^\dagger)^{-1}$, again a simple manipulation shows that $\tilde\Lambda_l=\Lambda_l^\uparrow$ for all $l$. Then we get $\tilde\Lambda_l(S_\Lc X^{-1}-X^\dagger)=0$ for all $l$, and, recalling that $\sum_{l=0}^\infty \Lambda_l^\dagger\tilde\Lambda_l=\1$, it follows that $S_\Lc X^{-1}-X^\dagger=0$. This again implies that $\Phi_{\Lc^\uparrow}(z,w)=\Phi_{\tilde\Lc}(z,w)$ and, more than this, also that $a_\Lc=a_{\Lc^\uparrow}$.

The conclusion is that, even if  we were suggested by the structure of the system to introduce two different dual sets of $\Lc$, they really collapse in just one set: needles to say, this is reminiscent of the existence of an unique Riesz basis which is bi-orthogonal to one given Riesz basis.

\vspace{1cm}

{\bf Remark:--} we have already mentioned the existence of several possible definitions of coherent states. In the so-called non linear states the factor $k!$ in the denominator is replaced by a different sequence. Here, of course, the same could be done by replacing in  definition (\ref{51})  $k!$ and $l!$ with $x_k!$ and $y_l!$, where $x_0!=y_0!=1$ and $x_k!=x_k\,x_{k-1}$ and $y_k!=y_k\,y_{k-1}$, for all $k\geq 1$, getting a two dimensional version of the non-linear coherent states. This extension is straightforward and will not be considered here.

\end{document}